\documentclass[preprint,12pt]{elsarticle}

\usepackage{graphicx}
\usepackage{xcolor} 
\usepackage{amsmath}
\usepackage{amssymb}
\usepackage{amssymb}
\usepackage{amsthm}
\usepackage{forloop} 
\usepackage{fp}
\usepackage{complexity}
\usepackage{subcaption}

\usepackage{etoolbox}
\makeatletter
\patchcmd{\ps@pprintTitle}
  {Preprint submitted}
  {Preprint accepted}
  {}{}
\makeatother

\newcommand\mycom[2]{\genfrac{}{}{0pt}{}{#1}{#2}}
\newtheorem{lemma}{Lemma}[section]
\newtheorem{corollary}{Corollary}[section]

\newtheorem{conjecture}{Conjecture}[section]
\newtheorem{theorem}{Theorem}[section]

\journal{Theoretical Computer Science}

\title{On the Approximability of Robust Network Design}

\author[1,2]{Yacine Al-Najjar\corref{cor1}}
\ead{yacine.alnajjar@huawei.com}

\author[2]{Walid Ben-Ameur}
\ead{walid.benameur@telecom-sudparis.eu}

\author[1]{J\'er\'emie Leguay}
\ead{jeremie.leguay@huawei.com}

\address[1]{Huawei Technologies, Paris Research Center, France.}
\address[2]{Samovar, Telecom SudParis, Institut Polytechnique de Paris, France.}

\cortext[cor1]{Corresponding author. Full postal address: \\ 18 quai du point du jour, 92100 Boulogne-Billancourt, France.}

\begin{document}

\begin{abstract}
{\color{black} Given the dynamic nature of traffic, we investigate the variant of robust network design where we have to determine the  capacity to reserve on each link   so that 
each demand vector belonging to a polyhedral set can be routed. The objective is either to minimize congestion or a linear cost. Routing is assumed to be fractional and dynamic (i.e., dependent on the current traffic vector).} 
We first prove that the 
robust network design problem with minimum congestion cannot be approximated within any constant factor. 
 Then, using the ETH conjecture, we get a $\Omega(\frac{\log n}{\log \log n})$ lower bound for the approximability of this problem. 
 This implies that the well-known $O(\log n)$ approximation ratio established by  R\"{a}cke in 2008 is tight. 
 Using Lagrange relaxation, we obtain a new proof of the $O(\log n)$ approximation.  An important consequence of the Lagrange-based reduction and our inapproximability results is that the robust network design problem with linear reservation cost cannot be approximated within any constant ratio. This answers a long-standing open question of  Chekuri (2007). \textcolor{black}{We also give another proof of the result of Goyal\&al (2009) stating that the optimal linear cost under static routing can be $\Omega(\log n)$ more expensive than the cost obtained under dynamic routing.} 
  Finally,  we show that even if only two given paths are allowed for each commodity, the robust network design problem   with minimum congestion or linear cost is hard to approximate within some constant.
 \end{abstract}


\maketitle


\section{Introduction}

Network  optimization~\cite{wang2008overview, luna2006network} plays a crucial role for telecommunication operators since it permits to carefully invest in infrastructures, i.e. reduce capital expenditures. As Internet traffic is ever increasing, the network's capacity needs to be expanded through careful investments every year or even half-year.
However, the dynamic nature of the traffic due to ordinary daily fluctuations, long term evolution and unpredictable events requires to consider uncertainty on the traffic demand when dimensioning network resources.

Ideally, the network capacity should follow the demand. When the traffic demand can be precisely known, several approaches have been proposed to solve the capacitated network design problem using for instance decomposition methods and cutting planes~\cite{COSTA20051429, FRANGIONI20091229, Raack11}. But in practice, perfect knowledge of future traffic is not available at the time the decision needs to be taken. The dynamic nature of the traffic due to ordinary daily fluctuations, long term evolution and unpredictable events requires to consider uncertainty on traffic demands when dimensioning network resources. While overestimated traffic forecasts could be used to solve a deterministic optimization problem, it is likely to yield to a costly over-provisioning of the network capacities, which is not acceptable. Therefore, robust optimization under uncertainty sets is a must for the design of network capacities. In this context, our paper presents new approximability results on two tightly related variants of the robust network design problem, the minimization of either the congestion or a linear cost.

Let's consider an undirected graph $G = (V(G), E(G))$ representing a communication network. The traffic is characterized by a set of commodities $h \in \mathcal{H}$ associated to different node pairs. And the routing of a commodity can be represented by a flow $f^h \in \mathbb{R}^{E(G)}$ of intensity $d_h$. To take into account the changing nature of the demand, $d$ is assumed to be uncertain and more precisely to belong to a polyhedral set $\mathcal{D}$. 
The \emph{polyhedral model} was introduced in~\cite{kerivin03,Ben-Ameur2005} as an extension of the \emph{hose model}~\cite{Duffield,FINGERHUT1997287}, where limits on the total traffic going into (resp. out of) a node are considered. 

When solving a robust network design problem, several objective functions can be considered. 
Given a capacity $c_e$ for each edge $e$, one might be interested in minimizing the congestion given by $\max_{e \in E(G) } \frac{ u_e }{ c_e }$ where $u_e$ is the reserved capacity on edge $e$. 
Another common objective function  is given by the linear reservation cost $ \sum_{e \in E(G)} \lambda_e u_e$. 
 This can also represent the average congestion by taking $\lambda_e = \frac{1}{c_e}$.
The goal is to choose a reservation vector $u$ so that the network is able to support any demand vector $d \in \mathcal{D}$, i.e.,  there exists a (fractional) routing serving every commodity such that the total flow on each edge $e$ is less than the reservation $u_e$.

{\color{black}The Robust network design variant that we are focusing on in this paper,} is  referred to as  \emph{dynamic routing} in the literature since the network is optimized such that any realization of  traffic matrix in the uncertainty set has its own routing. 
The robust network design problem where a linear reservation cost is minimized was proved to be  co-NP hard in~\cite{Gupta01} when the graph is directed. A stronger co-NP hardness result is given in~\cite{hardnessChekuri} where the graph is undirected (this implies the directed case result). Some exact solution  methods for robust network design have been considered in~\cite{KKT15,Mattia13}. 
\textcolor{black}{Some special cases where dynamic routing is easy to compute have been described in \cite{BENAMEURMultistatic2010Longer,MINOUX2010597,FRANGIONI201136}.}

Routing with uncertain demands has received a significant interest from the community. As opposed to dynamic routing, \emph{static routing} or \emph{stable routing} was introduced in~\cite{kerivin03}: it consists in choosing a fixed flow $x^h$ of value $1$ for each commodity $h$. The actual flow $f^h(d)$ for the demand scenario $d$ will then be scaled by the actual demand $d_h$ of commodity $h$, i.e. $f^h(d) = d_h x^h$.
Static routing is also called  \emph{oblivious routing} in~\cite{Applegate,Azar}. In this case, polynomial-time algorithms to compute optimal static routing (with respect to either congestion or linear reservation cost) have been proposed~\cite{kerivin03,Ben-Ameur2005,Applegate,Azar} based on either duality or cutting-plane algorithms.

To further improve solutions of static routing and overcome complexity issues related to \textit{dynamic routing}, a number of restrictions on routing have been considered to design polynomial-time algorithms (see~\cite{surveyWalid, Poss2014} for a complete survey). This includes, for example, the multi-static approach, introduced in~\cite{Ben-Ameur2007Multistatic},
where the uncertainty set is partitioned using an hyperplane and routing is restricted to be static over each partition.  
This idea has been generalized in~\cite{SCUTELLA2009197} to unrestricted covers of the uncertainty set and an extension to share the demand between routing templates, called volume routing, has been proposed in~\cite{Zotkiewicz2009}. \cite{Ouorou2007} applied affine routing for robust network design, based on affine adjustable robust counterparts introduced in~\cite{ben2004adjustable}, restricting the recourse to be an affine function of the uncertainties. The performance of this framework has been extensively compared to the static and dynamic routing, both theoretically and empirically~\cite{poss2013,Poss2014}. In practice, affine routing provides a good approximation of the dynamic routing while it can be solved in reasonable time thanks to polynomial-time algorithms. Finally,  an approach encompassing the previous approaches is the multipolar approach proposed in~\cite{multipolarRouting, multipolarRobustOptimization}.

In this work, we will only focus on the complexity of the robust network design problem  under dynamic routing, while minimizing either congestion or some linear cost. 
To close this section, let us  summarize the main contributions of the paper  and review some related work.

\subsection{Our contributions}

\begin{itemize}
\item We first prove that the robust network design problem with minimum congestion cannot be approximated within any constant factor. The reduction is based on the PCP theorem and some connections with the Gap-$3$-SAT problem~\cite{bookArora}. The same reduction also allows to show inapproximability within $\Omega(\log \frac{ n }{ \Delta } )$ where $\Delta$ is the maximum degree in the graph and $n$ is the number of vertices.  

\item Using the ETH conjecture~\cite{ethImpagliazzo1,ethImpagliazzo2}, we prove a $\Omega(\frac{\log n}{\log \log n})$ lower bound for the approximability of the robust network design problem with minimum congestion. This implies that the well-known $O(\log n)$ approximation ratio that can be obtained using the result in~\cite{obliviousRackeH2008} is tight.

\item We show that any $\alpha$-approximation algorithm for the robust network design problem with linear costs directly leads to an $\alpha$-approximation for the problem with minimum congestion. The proof is based on  Lagrange relaxation.  
We obtain that robust network design with minimum congestion can  be approximated within $O(\log n)$. 
This was already proved in~\cite{obliviousRackeH2008} in a different way. 

\item An important consequence of the Lagrange-based reduction and our inapproximability results is that the robust network design problem with linear reservation cost cannot be approximated within any constant ratio. This answers a long-standing open question stated in~\cite{Chekuri2007survey}.

\item \textcolor{black}{Another consequence is a new proof for the existence of instances for which the optimal static solution can be $\Omega(\log n)$ more expensive than a solution based on dynamic routing, when a linear cost is minimized. This was already proved in \cite{goyal2009} in a different way.} 

\item We show that even if only two given paths are allowed for each commodity, there is a constant $k$ such that the robust network design problem with minimum congestion or linear costs cannot be approximated within $k$.  
\end{itemize}

\subsection{Related work}

{\color{black}
Let us first assume that the graph is undirected and a linear cost is minimized.  A result attributed to A. Gupta (\cite{Chekuri2007survey}, see also \cite{goyal2009} for a more detailed presentation)  leads to an $O(\log n)$ approximation algorithm for linear cost under dynamic fractional routing.  Furthermore, this approximation is achieved by a routing on a (fixed) single tree. In particular, this shows that the ratios between the dynamic and the static solutions under fractional routing ($ \frac{ {Lin}_{stat-frac} }{ {Lin}_{dyn-frac} } $)  ($Lin$ denotes here the optimal linear cost of the solution) and between single path  and fractional routing under the static model ($ \frac{ {Lin}_{stat-sing} }{ {Lin}_{stat-frac} } $) is in $O( \log n)$ and provides  an $O(\log n)$ approximation for static single path routing ${Lin}_{stat-sing}$. On the other hand \cite{Olver2014} shows that the static single path problem cannot be approximated within a $\Omega( \log^{ \frac{1}{4} - \epsilon } n )$ ratio unless $NP \not\subset ZPTIME( n^{ \polylog (n)} )$. As noticed in \cite{goyal2009}, this implies (assuming this complexity conjecture) that the gap $\frac{ {Lin}_{stat-sing} }{ {Lin}_{stat-frac} } $ is in $\Omega( \log^{ \frac{1}{4} - \epsilon } n )$. \cite{goyal2009} has shown that the gap $ \frac{ {Lin}_{stat-frac} }{ {Lin}_{dyn-frac} } $  is $\Omega( \log n )$.}

{\color{black}For the linear cost and undirected graphs, an extensively studied polyhedron is the \emph{symmetric} hose model.  The demand vector is here not oriented  (i.e, there is no distinction between a demand from $i$ to $j$ and a demand from $j$ to $i$), and uncertainty is defined by considering an upper-bound limit $b_i$ for the sum of demands  related to node $i$.  
 A 2-approximation has been found for the dynamic fractional case~\cite{FINGERHUT1997287,Gupta01} based on tree routing (where we route through a static tree that should be found)  showing that $\frac{ {Lin}_{stat-tree} }{ {Lin}_{dyn-frac} } \leq 2$. It has been conjectured that this solution resulted in an optimal solution for the static single path routing. This question has been open for some time and has become known as the \emph{VPN conjecture}. It was finally answered by the affirmative in \cite{goyal2008}. 
 The \emph{assymetric} hose polytope was also considered in many papers.  An approximation algorithm  is proposed to compute ${Lin}_{stat-sing}$ within a ratio of  $3.39$  \cite{EIS}
  (or more precisely $2$ plus the best approximation ratio for the Steiner tree problem). If $\mathcal{D}$ is a \emph{balanced}
  asymmetric hose polytope, i.e., $\sum_{v \in  V} b^{out}_v = \sum_{v \in  V} b^{in}_v$ where $b^{in}_v$ (resp. $b^{out}_v$) is the upper bound for the traffic entering into (resp. going out of) $v$, then the best approximation factor becomes $2$ \cite{EIS}.
  Moreover, if we assume that $b^{out}_v = b^{in}_b$, then ${Lin}_{stat-sing}$ is easy to compute and we get that ${Lin}_{stat-tree} = {Lin}_{stat-sing}$ \cite{OLVER10}. In other words, there is some similarity with the case where $\mathcal{D}$ is a symmetric hose polytope.
}

{\color{black} When congestion is considered,}
\cite{Racke2002} proved the existence of  an oblivious (or static) routing with a competitive ratio of $O(\log^3 n)$ with respect to optimum routing of any traffic matrix.
Then,~\cite{HAR03} improved the bound to $O( \log^2 n \log \log n )$ and gave a polynomial-time algorithm to find such a static routing. Finally,~\cite{obliviousRackeH2008} described an $O(\log n)$ approximation algorithm for static routing with minimum congestion.  
Notice that the bound given by static routing cannot provide a better bound than $O(\log n)$ since  a lower bound of  $\Omega(\log n)$ is achieved by static routing for planar graphs~\cite{Maggs97,BARTAL199919}.
{\color{black}
It has also been shown  in \cite{Hajiaghayi2007polyPath} that the gap between the dynamic fractional routing and a dynamic fractional routing restricted to a  polynomial number of paths can be $\Omega( \frac{\log n}{ \log \log n} )$. 

When a directed graph is considered and congestion is minimized,  \cite{Azar} has shown that the gap between static fractional routing and  dynamic fractional routing can be $\Omega(\sqrt{n})$ while \cite{Hajiaghayi2007}  proves that the gap is upper-bounded by $O(\sqrt{k} n^{\frac{1}{4}} \log n)$  (where $k = |\mathcal{H}|$ is the number of commodities). More results can be found in \cite{Hajiaghayi2007} and the references therein.  
}

Using an approximate separation oracle for the dual problem to obtain an approximate solution of the primal is a well-known technique already used in~\cite{Fleischer06,Carr02,Jain03} at least in the context of packing-covering problems. Lagrangian relaxations are also used in    
\cite{Garg1998FasterAS,Plotkin1991FastAA,Young95} to produce dual solutions that are near-optimal.

\section{From Gap-3-SAT to robust network design with minimum congestion}

\label{sec:reduction}
Given an edge $e$, let $s(e)$ and $t(e)$ be the extremities of $e$.   
Similarly to edges,  for a commodity $h \in \mathcal{H}$,  let $s(h)$  and $t(h)$ denote the endpoints of  $h $. 
And let $\mathcal{U}( \mathcal{D} )$ be the set of  $u \in \mathbb{R}^{E(G)} $ such that each traffic vector $d \in \mathcal{D}$ can be routed on the network when a capacity \textcolor{black}{$u_e$} is assigned to edge $e$. Since $\mathcal{D}$ is polyhedral, $\mathcal{U}( \mathcal{D} )$ is also polyhedral (see, e.g,~\cite{Chekuri2007survey}).

We are interested in minimizing the congestion under  polyhedral uncertainty and dynamic routing: $ \min\limits_{ u \in \mathcal{U}( \mathcal{D} ) } \max\limits_{e \in E(G) } \frac{ u_e }{ c_e }$.

Given a polytope represented by  $A x \leq b$, the size of the polytope denotes the total encoding size of the entries in $A$ and $b$.   

Our first main result is related to the inapproximability of the minimum congestion problem within a constant factor. 

\begin{theorem}
\label{inapproximabilityPNP}
 Unless $P = NP$, the minimum congestion problem cannot be approximated with a polynomial-time algorithm within any constant factor even if $\mathcal{D}$ is given by $\{d: A d + B {\color{black} \psi } \leq b \}$ whose size  is polynomially bounded by $|V(G)|$. 
\end{theorem}
Notice that it is important to consider polyhedral uncertainty sets that are easy to describe (otherwise the inapproximability results would be a direct consequence of the difficulty to separate from the uncertainty set). 

To prove Theorem~\ref{inapproximabilityPNP}, we will need the PCP (Probabilistically Checkable Proof) theorem~\cite{bookArora} and an intermediate lemma. 
For a 3-SAT formula $\varphi$ we note $val( \varphi )$ the maximum fraction of the clauses which are satisfiable at the same time. In particular, $val( \varphi ) = 1$ means that $\varphi$ is satisfiable.
The problem where we have to decide if $val(\varphi) < \rho$ or $val(\varphi) = 1 $ for a 3-SAT formula $\varphi$ is called Gap-3-SAT. 
{\color{black}
The instances such that $\rho \leq val(\varphi) < 1$ do not need to be considered. 
One way to state PCP theorem is to say that there exists a constant $0 < \rho < 1$  for which Gap-3-SAT is NP-hard.  In other words, it is NP-hard to distinguish between satisfiable 3-SAT formulas and those for which strictly less than a fraction $\rho$ of clauses can be simultaneously satisfied. 
}

To prove the theorem \ref{inapproximabilityPNP}, we will use the following lemma (where \emph{cong} denotes the optimal congestion of the corresponding instance).

\begin{lemma}
\label{lemmaInapproximability}
 For every $\gamma \in \mathbb{N}$ there is a mapping $f_\gamma$ computable in polynomial time from 3-SAT instances to minimum congestion instances defined by an undirected graph  $G_\gamma$, a set of commodities $\mathcal{H}_\gamma$ and a polytope $\mathcal{D}_\gamma = \{d: A_\gamma d + B_\gamma \psi_\gamma \leq b_\gamma \}$  such that   $|V(G_\gamma)| = O(m^{\gamma})$, $|E(G_\gamma)| = O(m^{\gamma})$ and the size of $\mathcal{D}_\gamma$ is $O(m^{c \gamma})$ where $c$ is some positive constant and $m$ is the number of clauses. The mapping satisfies the following:  
\begin{itemize}
    \item $ val(\varphi) = 1 \implies  cong(f_\gamma(\varphi) ) \geq 1 + \gamma (1 -\rho)  $
    \item $val(\varphi) < \rho  \implies cong(f_\gamma(\varphi) ) \leq 1$.
\end{itemize}
\end{lemma}

\begin{proof}{ \emph{ of Theorem \ref{inapproximabilityPNP}  }  }
 We are going to use Lemma \ref{lemmaInapproximability} and PCP Theorem for the proof.
Suppose that congestion can be approximated in polynomial time within a constant approximation factor $\alpha$. We first choose $\gamma$ such that $\alpha < 1 + \gamma(1 -\rho) $.

{\color{black} Starting from a 3-SAT formula such that either $val(\varphi) < \rho$ or $val(\varphi) = 1$, we construct $f_\gamma ( \varphi )$ in polynomial time. The optimal congestion will satisfy either $cong(f_\gamma(\varphi) ) \geq 1 + \gamma (1 -\rho)  $ or  $ cong(f_\gamma(\varphi) )  \leq 1$. 
 Applying the  $\alpha$-approximation to $f_\gamma ( \varphi )$ provides an  approximate value $\Tilde{\beta}$ for congestion. 
If $\Tilde{\beta} < 1 + \gamma(1 -\rho)$ holds, then we can deduce that $ cong(f_\gamma(\varphi) ) \leq \Tilde{\beta} < 1 + \gamma(1 -\rho)$. This implies that
$cong(f_\gamma(\varphi)  \leq 1$ and hence
$val(\varphi ) < \rho$.
Otherwise, we have  $\Tilde{\beta} \geq 1 + \gamma(1 -\rho) $ and  $\alpha \times cong(f_{\gamma }( \varphi) ) \geq \Tilde{\beta}$ (since $\Tilde{\beta}$ is an $\alpha$-approximation), leading to  $cong(f_{\gamma }( \varphi) ) \geq \frac{1 + \gamma(1 -\rho)}{ \alpha} > 1$.  We consequently have $cong(f_{\gamma }( \varphi) ) \geq  1 + \gamma(1 -\rho) $ and  
 $val(\varphi ) = 1$. This proves that a constant $\alpha$-approximation for the congestion problem allows the solution of Gap-3-SAT. 
}
Furthermore, as the size of the polytope used in Lemma~\ref{lemmaInapproximability} is $O(m^{c \gamma})$ while $|V(G_\gamma)| = O(m^{\gamma})$, its size is polynomially bounded in the number of vertices as announced in Theorem~\ref{inapproximabilityPNP}
\end{proof}

We are now going to prove Lemma \ref{lemmaInapproximability} by  first constructing instances of the congestion problem leading to some inapproximabilty factor. Then, this factor is increased by recursively building larger instances with higher values of $\gamma$.

\begin{proof}{ \emph{of Lemma \ref{lemmaInapproximability}, case $\gamma = 1$ }  }
\begin{figure}
\center
\includegraphics[width=13.5cm,height = 3.5cm]{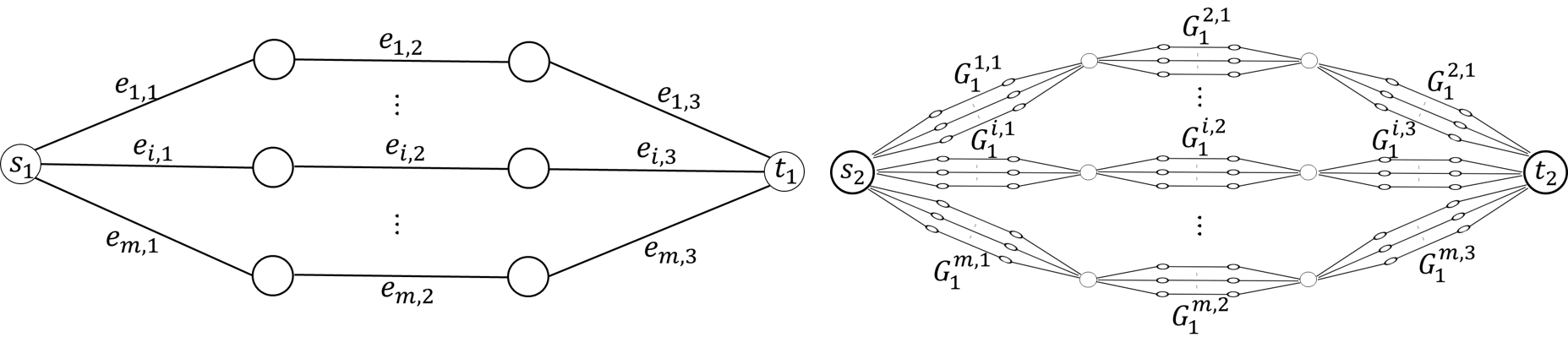}
\caption{$G_1$ and $G_2$}
\label{fig:g1}
\end{figure}

We start with a 3-SAT formula $\varphi$, with $m$ clauses and $r$ variables. We note $\mathcal{L} = \{l_1, \ldots, l_r,\lnot l_1, \ldots, \lnot l_r \}$ the set of the literals appearing in  formula $\varphi$ and $l_{i,j}$ the literal appearing in the $i$-th clause $C_i$ at the $j$-th position for $i=1,..., m$ and $j=1,2,3$ (it is not restrictive to assume that each clause contains exactly $3$ literals). 

We build as follows a graph  $G_1$ and  a set of commodities $\mathcal{H}_1$. 
For each $i = 1, ..., m$, $j = 1,2,3$ we add 3 consecutive edges $e_{i,j}$ (i.e. such that $t(e_{i,1} ) = s( e_{i,2 } )$ and $t(e_{i,2} ) = s( e_{i,3 })$)  and 3 commodities $h_{i,j}$ with $s( h_{i,j} ) = s( e_{i,j} ) $ and $ t( h_{i,j} ) = t( e_{i,j} ) $.  We impose that all nodes $s(e_{i,1})$ (resp. $t(e_{i,3})$) for $i=1,...,m$ are  equal to a single node noted $s_1$ (resp. $t_1$) (see Figure \ref{fig:g1}).
 We consider an additional commodity $h_0$ between $s_1$ and $t_1$. 
{\color{black}
We create a polyhedron $\mathcal{D}_1$ of the form $\mathcal{D}_1 = \{d: A_1 d + B_1 \psi_1 \leq b_1 \} $ as follows. We  consider for each literal $l \in  \mathcal{L}$  a non-negative variable $\xi_l$ and add  for $k = 1, ... , r$ the constraint
$\xi_{l_k} + \xi_{\lnot l_k} = 1$. We also consider, for $i=1,...,m$ and $j=1,2,3$, the constraint $d_{h_{i,j} } = \xi_{l_{i,j}}$. A constraint related to $d_{h_0}$ is also integrated: 
 $d_{h_0} \leq m (1 - \rho)$.} 
 Finally, the capacity $c_e$ of each  edge $e$ is here equal to $1$ ($c_e = 1$).

 If $val(\varphi ) = 1$,   then there is a demand vector such that 
  for each path between $s_1$ and $t_1$ (there is one path corresponding to each clause),  at least one commodity  whose endpoints are on the path  is equal to $1$ (a commodity corresponding to a true literal).  This implies that all paths are blocked and thus the optimal routing for commodity $h_0$ is to equally spread $m (1 - \rho)$  between the $m$ paths leading to a congestion  of $1 + (1 - \rho )$.

  Let us now assume that $val(\varphi ) < \rho$. 
  Notice first that the components of the extreme points of the polyhedron $\mathcal{D}_1$  are integers (except $d_{h_0}$). This is due to the fact that $\mathcal{D}_1$ can be seen as a coordinate projection of the higher dimensional polytope $\{(d,\psi) | A_1 d + B_1 \psi \}$ whose extreme points are obviously integers (except $d_{h_0}$). 
For such an extreme  demand vector $d \in \mathcal{D}_1$ there are at least $m (1 - \rho)$ free paths to route the demand $d_{h_0}$ allowing a congestion less than or equal to $1$. This implies that all demands in $\mathcal{D}_1$ can also be routed with a congestion less than or equal to $1$

Observe that  $|V(G_1)| = O(m)$, $|E(G_1)| = O(m)$, $\mathcal{D}_1$ has the appropriate form ($\mathcal{D}_1 = \{d: A_1 d + B_1 \psi_1 \leq b_1 \}$) and the size of $\mathcal{D}_1$ is $O(m^{c})$ for some constant $c$. 
 \end{proof}

\begin{proof}{ \emph{of Lemma \ref{lemmaInapproximability}, case $\gamma \geq  2$ } }

For $\gamma\geq 2$, having constructed $G_{\gamma-1} , \mathcal{H}_{\gamma-1} , \mathcal{D}_{\gamma -1} $, we build $G_{\gamma} , \mathcal{H}_{\gamma} , \mathcal{D}_{\gamma} $  as follows.
 We will construct the graph $G_\gamma$, by taking the graph $G_1$ and replacing each edge by a copy of the graph $G_{ \gamma-1 }$ denoted by $G^{{i,j} }_{\gamma-1}$. Each copy $G^{{i,j} }_{\gamma-1}$ contains a node $s_{\gamma-1}$  that is identified with $s(e_{i,j})$  and a  node $t_{\gamma-1}$ identified with $t(e_{i,j})$ (see Figure \ref{fig:g1}). 
All commodities related to $G^{{i,j} }_{\gamma-1}$ (belonging to $ \mathcal{H}_{\gamma-1} $) are also considered as commodities of $\mathcal{H}_{\gamma}$. Let us use $d^{i,j} \in \mathbb{R}^{\mathcal{H}_{\gamma-1} }$
to denote the related demand vector. 
$\mathcal{H}_{\gamma}$  also contains a non-negative commodity $h_{0,\gamma} $ constrained by $d_{h_{0,\gamma} } \leq  m^\gamma (1 - \rho)$. Thus $|\mathcal{H}_{\gamma}| = 1 + 3m \times |\mathcal{H}_{\gamma - 1}|$. 

{\color{black}{We are going to build an uncertainty set $\mathcal{D}_{\gamma}$  as a coordinate projection of a higher-dimensional polyhedron  $\Xi_\gamma$, involving  demand variables in addition to auxiliary non-negative variables $\xi_{l}$ related to literals, and also auxiliary variables $\psi^{i,j}_{\gamma - 1}$ related to $G^{{i,j} }_{\gamma-1}$ and the description of $\mathcal{H}_{\gamma-1}$. We gradually explain the construction.} 
For $k = 1, ..., r$, we add the constraint $\xi_{l_k} + \xi_{\lnot l_k} = 1$. 
And for  $e_{i,j} \in E(G_{ 1 } )$, 
we impose that $d^{i,j} \in  \xi_{l_{i,j}} \mathcal{D}_{\gamma-1}$  $:=\{ \xi_{l_{ij}} d_0 | d_0 \in \mathcal{D}_{\gamma - 1 } \} $}.  
Let us explain how this can be done. By induction, we know that $\mathcal{D}_{\gamma-1} = \{d: A_{\gamma -1 }d + B_{\gamma - 1} \psi_{\gamma - 1} \leq b_{\gamma - 1} \}$ and this representation includes (among others) non-negativity constraints of all variables in addition to constraints implying that all variables are upper-bounded.  Then by writing $A_{\gamma -1 } d^{i,j} + B_{\gamma - 1} \psi^{i,j}  \leq \xi_{l_{i,j}} b_{\gamma - 1}$, we can ensure that  $\xi_{l_{i,j}} = 0$ implies $d^{i,j}=0$, while  $\xi_{l_{i,j}} > 0$ leads to $\frac{1}{\xi_{l_{i,j}}} d^{i,j} \in \mathcal{D}_{\gamma-1}$.  In particular when $\xi_{l_{i,j}} =0$,
from  outside,  the whole subgraph corresponding to  $G^{i,j}_{\gamma-1}$ acts like a single edge of capacity $m^{\gamma-1}$.

$\mathcal{D}_{\gamma}$ can be seen as the projection of a polytope $\Xi_{\gamma} = \{ (d,\psi_\gamma ) | A_{\gamma} d + B_\gamma \psi_\gamma \leq b_\gamma \}$ where $\psi_\gamma$ contains the auxiliary variables  appearing in  all levels. More precisely, $\Xi_{\gamma}$ is defined by: 
\begin{align}
  d_{h_{0,\gamma} } & \leq   m^\gamma (1 - \rho) \nonumber  \\
 - d_{h_{0,\gamma} } & \leq  0, \quad \nonumber \\ 
 - \xi_l  & \leq 0, \quad  \forall  l \in \mathcal{L} \nonumber  \\
  \xi_{l_k} + \xi_{\lnot l_k} & \leq  1, \quad \forall k = 1, ..., r \nonumber \\
 -\xi_{l_k} - \xi_{\lnot l_k} & \leq   -1, \quad
 \forall k = 1, ..., r \nonumber \\
 A_{\gamma - 1} d^{i,j} + B_{\gamma - 1} \psi^{i,j}_{\gamma - 1 } - \xi_{l_{i,j}} b_{\gamma - 1} & \leq  0, \quad    \forall i = 1,...,m, j = 1, 2, 3. \label{eq:induc} 
\end{align}

By simple induction, we have  $|V(G_\gamma)| = O(m^{\gamma})$, $|E(G_\gamma)| = O(m^{\gamma})$ and the size of $\mathcal{D}_\gamma$ is $O(m^{c \gamma})$ where $c$ is some positive constant. 

We observe that all extreme points of $\Xi_\gamma$  are such that  $\xi_l \in  \{0,1\}$ for $l \in \mathcal{L}$. To verify that, we first recall that  constraints \eqref{eq:induc} are equivalent to $d^{i,j} \in  \xi_{l_{i,j}} \mathcal{D}_{\gamma-1}$ (in this way, the vectors $\psi^{i,j}_{\gamma - 1 }$ can be ignored). 
Second, let $\mathcal{L_+}$ be the set of literals appearing in positive form.
We observe that variables $\xi_l$ for $l \in \mathcal{L_+}$ are pairwise independent. Only variables $d^{i,j}$  such that either $l_{i,j} = l$ or $l_{i,j} = \lnot l$ depend on $\xi_l$ since 
$d^{i,j}\in  \xi_{l} \mathcal{D}_{\gamma-1}$ in the first case and $d^{i,j} \in  (1 - \xi_{l}) \mathcal{D}_{\gamma-1}$ in the second case. This immediately implies that given some arbitrary real vectors $q_{i,j}$ and $f$, minimizing $\sum\limits_{i = 1,..,m;j = 1, 2, 3} q_{i,j}^T d^{i,j} +   \sum\limits_{l \in \mathcal{L_{+}}} f_l \xi_l$ is equivalent to minimizing  
$\sum\limits_{l \in \mathcal{L_{+}}} \xi_l \left(f_l + \sum\limits_{i,j: l_{i,j} = l } \min\limits_{d^{i,j}\in   \mathcal{D}_{\gamma-1}} q_{i,j}^T d^{i,j} - \sum\limits_{i,j: l_{i,j} = \lnot l } \min\limits_{d^{i,j}\in   \mathcal{D}_{\gamma-1}} q_{i,j}^T d^{i,j} \right)$. It is then clear that optimal $\xi_l$ values will be either $0$ or $1$.  Since this holds for an arbitrary linear objective function, we get the wanted result about extreme points. \\

Let us now show that $val(\varphi) < \rho  \implies cong(f_\gamma(\varphi) ) \leq 1$.
Assume that  $val(\varphi) < \rho$.  We prove by induction that the congestion of $(G_\gamma,\mathcal{H}_\gamma,\mathcal{D}_\gamma)$ is $1$. Suppose that this is true for some $\gamma-1$. If $\xi_{l_{i,1}} = \xi_{l_{i,2}} = \xi_{l_{i,3}} = 0$ for some $i$, a flow of value $m^{\gamma-1}$ can be routed between $s_\gamma$ and $t_\gamma$ by sending a flow of value 1 on each edge of $G^{{i,j} }_{\gamma-1}$  for $j=1,2,3$. Since $val(\varphi) < \rho$, there are necessarily at least $m(1-\rho )$ such $i$, thus we can send the whole demand $m^{\gamma-1} m(1-\rho ) = m^{\gamma}(1-\rho )$ this way. For the indices $i,j$ such that $\xi_{l_{i,j}} = 1$, by the induction  hypothesis ($cong(f_{\gamma-1}(\varphi) ) \leq 1$), the demands inside $G^{{i,j} }_{\gamma-1}$ can be routed without sending more than one unit of flow on each edge of $G^{{i,j} }_{\gamma-1}$. 

Notice that to show that all traffic vectors  of $\mathcal{D}_{\gamma}$ can be routed with congestion $1$, we considered demand vectors corresponding with $\{0,1\}$ $\xi$ variables. The result shown above  about extreme points is useful here since it allows us to say that each extreme point of $\mathcal{D}_{\gamma}$ can be routed with congestion less than or equal to $1$ implying that each demand vector inside $\mathcal{D}_{\gamma}$ can also be routed with congestion less than or equal to $1$. \\

Let us now show that
$val(\varphi) = 1 \implies  cong(f_\gamma(\varphi) ) \geq 1 + \gamma(1 -\rho)$.
We are going to use induction  to build  a cut $\delta(C_\gamma)$ where $C_\gamma$ is set of vertices of $V(G_\gamma)$ containing $s_\gamma$ and not containing $t_\gamma$. The number of edges of  the cut will be $m^\gamma$ and each edge has a capacity equal to $1$. We also show the existence of a demand vector $d \in \mathcal{D}_\gamma$   such that the sum of the demands traversing the cut is greater than or equal to $m^\gamma(1 + \gamma(1- \rho) )$. 
This would show that there is at least one edge that carries at least $1 + \gamma(1- \rho)$ units of flow. 

Since $\varphi$ is satisfiable, there is a truth assignment represented by $\xi$ variables  (the auxiliary variables) such that for each $i=1,...,m$ there is a $j(i)$ such that $\xi_{l_{i,j(i)}} = 1$.
 By considering the graph $G_{\gamma-1}^{{i,j(i)} }$ and using the induction hypothesis, we can build a cut $\delta(C^{i}_{\gamma-1})$ separating the node $s(e_{i,j(i)})$ and  $t(e_{i,j(i)})$   and  containing $m^{\gamma-1}$ edges. We also build a demand vector $d^{{i,j(i)}} \in \mathcal{D}_{\gamma-1}$ such that the sum of demands traversing the cut is greater than or equal to $m^{\gamma-1} (1 + (\gamma-1)(1-\rho ) )$ (still possible by induction). By taking the union of these $m$ disjoint cuts we get a cut $\delta(C_\gamma)$ that is separating $s_\gamma$ and $t_\gamma$ having the required number of edges. A demand vector $d$ can be built by combining the vectors $d^{{i,j(i)}}$ and the demand $d_{h_0, \gamma}$ taken equal to $m^{\gamma}(1 - \rho) $. Since the demand from $s_\gamma$ to $t_\gamma$ is also traversing the cut, the total demand through $\delta(C_\gamma)$ is greater than or equal to $m^\gamma (1 - \rho) + m . m^{\gamma-1} (1 + (\gamma-1)(1-\rho ) ) = m^\gamma( 1 + \gamma ( 1-\rho ) ) $.  
\end{proof}

Lemma \ref{lemmaInapproximability} can be further exploited in different ways since there are many possible connections between the value $1 + \gamma (1 - \rho)$ and the characteristics of the undirected graph built in the proof of the lemma.  Observe, for example, that by a simple induction we get that the number of vertices $|V(G_\gamma)| = 2 + 2m \frac{(3m)^{\gamma} - 1}{3m - 1}$ leading to $|V(G_\gamma)|  \simeq 2 \times 3^{\gamma - 1} m^{\gamma}$ (when $m$ goes to infinity). 
We also have $\Delta(G_\gamma)$ equal to $m^{\gamma }$ where $\Delta(.)$ denotes the maximum degree in the graph. Consequently, $log(\frac{|V(G_\gamma)|}{\Delta(G_\gamma)}) \simeq \gamma  \log 3  + \log 2/3 $. Then by taking any constant $k$ such that $k \times \log 3 < (1 - \rho)$ where $\rho$ is the constant in the PCP Theorem we get a lower bound of the approximability ratio. 
This is stated in the following corollary.
\begin{corollary}
\label{coro:vdelta}
Under  conditions of Theorem~\ref{inapproximabilityPNP}, for any constant $k < \frac{1 - \rho}{\log 3 }$, it is not possible to approximate the minimum congestion problem in polynomial time within a ratio of $k \log \frac{|V(G)|}{\Delta(G)} $.
\end{corollary}

  \section{ A $\Omega(\frac { \log n}{\log \log n}) $ approximability lower bound }
  
\textcolor{black}{ To get an approximability lower bound, we will use the well-known ETH conjecture that is recalled below.  }
\begin{conjecture}[Exponential Time Hypothesis]
\label{ETH}
\cite{ethImpagliazzo1,ethImpagliazzo2} There is a  constant $ \delta$ such that no algorithm can solve 3-SAT instances  in time $O( 2^{\delta  m })$, where m is the number of clauses.
\end{conjecture}
Let us use $n$ to denote the number of vertices of the graph.

\begin{theorem}
Under Conjecture \ref{ETH},  there exists a constant $k$ such that no polynomial-time algorithm can solve the minimum congestion problem with the approximation ratio $k \frac{\log  n }{\log \log n}$. 
\label{log/loglog}
\end{theorem}

\begin{proof}
The combination of PCP Theorem 
and ETH Conjecture \ref{ETH} implies that distinguishing between 3-SAT instances such that $val( \varphi ) < \rho $ and $val( \varphi ) = 1$ cannot be done in time $O( 2^{ m^{\beta } } )$ for some constant $\beta > 0$ (a slightly better bound is $O(2^{ {m}/{\log^c m }})$ for some constant $c$, but this will not help us to improve the lower bound of Theorem \ref{log/loglog}).

Suppose that there is an algorithm that solves the minimum congestion problem with an approximation factor $\alpha(n)$ and a running time $O(n^{ c_1 } )$.  Given a 3-SAT instance and a function $ \gamma : \mathbb{N} \xrightarrow{} \mathbb{N} $ we can construct a minimum congestion instance $ f_{ \gamma ( m ) }( \varphi )  $ as in Lemma \ref{lemmaInapproximability} in
time $O(m^{ c_2 \gamma(m) }) $ and where the number of vertices of the instance is
$m^{  \gamma(m) }$. Then by running the approximation algorithm for minimum congestion we get a total time of $ O( m^{ c_3 \gamma(m) } ) $ where $c_3 = \max\{c_1, c_2\}$. 
Thus by choosing  $\gamma(m) = \frac{m^\beta }{ c_3  \log  m  }$ we get an algorithm that runs in time $O( 2^{  m^{ \beta }      } )$. And if the approximation factor $\alpha(n)$ is small enough, that is if   $ \alpha ( m^{ \gamma(m)} ) < 1 + (1 - \rho) \gamma(m)$ for a big enough $m$, we get an algorithm solving Gap-3-SAT and thus contradicting Conjecture \ref{ETH}. This is the case for $k \frac{\log n }{\log \log n }$ for some constant $k$.  To see this, we can observe that:

$\frac{ 1 + (1 -\rho) \gamma(m) }{\alpha ( m^{  \gamma(m)} ) }
= \frac{ 1 + (1 -\rho) \frac{m^\beta }{ c_3 \log m } }{  k \frac{m^\beta / c_3}{\beta \log m - \log c_3 }  }  \simeq \frac{\beta(1 - \rho)}{k }$.
By taking $k <{\beta(1 - \rho)}$ we get the wanted inapproximability result. 
\end{proof}

\section{From minimum congestion to linear costs}
\label{sec:lagrange} 

Given any $ \lambda \geq 0$, the robust network design problem with linear costs is simply the following \textcolor{black}{where $\mathcal{U}(\mathcal{D})$ is the set of possible capacity vectors defined in Section \ref{sec:reduction}}:
\begin{equation}
\min\limits_{u \in \mathcal{U}(\mathcal{D} )  } \lambda^T u. 
\label{linearCost}
\end{equation}

Assume that there exists a number $\alpha \geq 1$ such that Problem \eqref{linearCost} can be solved in polynomial time within an approximation ratio $\alpha$.  More precisely, we have 
a polynomial-time oracle that takes as input a non-negative  linear cost $ \lambda \in \mathbb{R}^{ E(G) }$ and outputs a $ u^{ap}( \lambda )  \in \mathcal{U}( \mathcal{ D } )$
such that $ \lambda^T u(\lambda )  \leq \lambda^T u^{ap}(\lambda ) \leq \alpha \lambda^T u( \lambda ) $ where $ u(\lambda ) \in \mathcal{U}(\mathcal{D} )$ is the optimal solution of \eqref{linearCost}.

Recall that the congestion problem is given by
{\color{black}
\begin{align}
\label{congestionPb}
&\min\limits_{   \beta , u } \beta \\
 & u_e \leq c_e \beta,  \forall e \in E(G) \nonumber   \\
 & u \in \mathcal{U}(\mathcal{D}) \nonumber
\end{align}
where $\beta$ and $u$ are optimization variables.
}

Let us consider a Lagrange relaxation of \eqref{congestionPb} by dualizing the capacity constraints and using $\lambda$ for the dual multipliers.  The dual problem is then given by $ \max\limits_{\lambda \geq 0} \min\limits_{\beta , u \in \mathcal{U}(\mathcal{D} )  } \beta + \sum_{e \in E(G)} \lambda_e (u_e - \beta c_e) $ (where $\beta$ is an optimization variable). 
{\color{black} If $\lambda$ is chosen such $\sum\limits_{e} \lambda_e c_e \neq 1 $, then the value of the inner minimum  would be $-\infty$. Thus in an optimal solution, we will always have $\sum\limits_{e} \lambda_e c_e = 1 $. The problem is then equivalent to:  }

\begin{equation}
\label{dualPb}
\max\limits_{\mycom{\lambda \geq 0}{\sum\limits_{e \in E(G)} \lambda_e c_e = 1 }} \min\limits_{u \in \mathcal{U}(\mathcal{D} )  } \sum_{e \in E(G)} \lambda_e u_e =  \max\limits_{\mycom{\lambda \geq 0}{\sum\limits_{e \in E(G)} \lambda_e c_e = 1 }} \lambda^T u( \lambda ).  
\end{equation}
Since $\mathcal{U}(\mathcal{D} )$ is polyhedral and all constraints and the objective function are linear, there is 
 is no duality gap between \eqref{congestionPb} and \eqref{dualPb}. 

Observe that \eqref{dualPb} can be expressed as follows: 
\begin{subequations}
\label{suite:cong}
\label{congestionDual}
\begin{align}
 & \max\limits_{ \beta, \lambda \geq 0 } \beta \\
 \beta & \leq \sum\limits_{ e \in E(G) }   \lambda_e u_e , \forall u \in  \mathcal{U} ( \mathcal{D} )   \label{cont:dual} \\
  1 &=  \sum\limits_{ e \in E(G) } \lambda_e c_e
\end{align}
\end{subequations}

We are going to approximately solve \eqref{congestionDual} using a cutting-plane algorithm where inequalities \eqref{cont:dual} are iteratively added by using the $\alpha$-approximation oracle. 
Let $ ( { \beta }' , {\lambda }' )$ be a potential solution of (\ref{congestionDual}), we can run the $\alpha$-approximation of robust network design problem (\ref{linearCost}) with the cost vector ${\lambda}'$ to get a solution ${ u^{ap}(\lambda') }$. 
If $ {\beta}' > \sum\limits_{ e \in E(G) }  {\lambda}'_e  u^{ap}_e(\lambda') $ we return the inequality $ {\beta}  \leq  \sum\limits_{ e \in E(G) }  {\lambda}_e  u^{ap}_e(\lambda') $,  otherwise the algorithm stops and returns $ ( { \beta }' , {\lambda }' )$. 
We know from the separation-optimization equivalence theorem~\cite{GroetschelLovaszSchrijver1988}  that  \eqref{congestionDual} can be solved by making a polynomial number of calls to the separation oracle leading a globally polynomial-time algorithm. Notice that this happens if the separation oracle is exact. In our case, the oracle is only an approximate one, implying that the cutting plane algorithm might be   prematurely  interrupted before obtaining the true optimum of \eqref{congestionDual}. Observe however that this implies that the computing time is polynomially bounded. 
Let  $ ( \Tilde{ \beta } , \Tilde{\lambda })$ be the solution returned by the cutting-plane algorithm.   Let  $( { \beta }^* , {\lambda }^*)$ be the true optimal solution of \eqref{congestionDual}.  The next lemma states that the returned solution is an $\alpha$-approximation of the optimal solution.
\begin{lemma} The cutting-plane algorithm computes in polynomial time a solution $\Tilde{ \beta }$ satisfying: 
\begin{equation}
\beta^* \leq \Tilde{ \beta } \leq \alpha \beta^*.
\label{eq:approx}
\end{equation}
\label{lem:approx}
\end{lemma}
\begin{proof}
Observe that $\beta^* = {\lambda^*}^T  u(\lambda^*)$. Moreover, since \eqref{congestionDual} is equivalent to \eqref{dualPb}, we get that
$
{\lambda^*}^T  u(\lambda^*) =  \beta^* \geq \Tilde{\lambda}^T u( \Tilde{\lambda} )
$.
From the approximation factor of the oracle, one can write that  
$\Tilde{\lambda}^T u^{ap }( \Tilde{\lambda} ) \leq \alpha \Tilde{\lambda}^T  u(\Tilde{\lambda}).
$
Using the fact that no inequalities can be added for   $ ( \Tilde{ \beta } , \Tilde{\lambda })$, we get that
$
\Tilde{\beta}  \leq  \Tilde{\lambda}^T  u^{ap }( \Tilde{\lambda} )$.
Finally, since $(\beta^*,\lambda^*)$ is feasible for  \eqref{congestionDual}, we obviously have
$\Tilde{\beta} \geq  \beta^*   $.
Combining the $4$ previous inequalities leads to \eqref{eq:approx}. 
\end{proof}

The above lemma has many consequences. 
\begin{theorem}
{\color{black} Unless $P=NP$,} the robust network design problem  with linear costs cannot be approximated {\color{black} in polynomial time} within any constant ratio. {\color{black} Unless the ETH conjecture is false, the robust network design problem  with linear costs cannot be approximated within  $\Omega( \frac{\log n}{ \log \log n} )$. }
\end{theorem}
\begin{proof}
The result is an immediate consequence of Theorem{\color{black}s} \ref{inapproximabilityPNP}{\color{black}, \ref{log/loglog}} and Lemma \ref{lem:approx}.
\end{proof}
The theorem above answers a long-standing open question of~\cite{Chekuri2007survey}. All other inapproximability results  proved for the congestion problem directly hold for the robust network design problem with linear cost.  

Another important consequence is that the congestion problem can be approximated within $O(\log n)$.  This result was already proved in~\cite{obliviousRackeH2008}  using other techniques. In our case, the result is an immediate consequence of the $O(\log n)$-approximation algorithm for the robust network design problem with linear cost provided by~\cite{personalCommunication2004,GUPTA20113} and fully described in~\cite{Chekuri2007survey,goyal2009}. 
\begin{theorem}~\cite{obliviousRackeH2008} 
 Congestion  can be approximated within $O(\log n)$.
\end{theorem}
Notice that Theorem \ref{log/loglog} tells us that the ratio  $O(\log n)$ is tight.

 \textcolor{black}{Starting from the results of  \cite{Maggs97,BARTAL199919} showing the existence of 
 instances for which the ratio $\frac{cong_{stat-frac}}{cong_{dyn-frac}}$ is $\Omega(\log n)$, one can also use the reduction above to prove, differently from \cite{goyal2009}, the existence of instances for which the ratio $\frac{{Lin}_{stat-frac}}{{Lin}_{dyn-frac}}$ is $\Omega(\log n)$ where a linear cost is minimized.  
 \begin{theorem}~\cite{goyal2009} 
 There are instances for which $\frac{{Lin}_{stat-frac}}{{Lin}_{dyn-frac}}$ is $\Omega(\log n)$.
 \label{the:gap}
\end{theorem}
 \begin{proof} Similarly to $\mathcal{U}(\mathcal{D})$ defined when dynamic routing is considered, let  $\mathcal{U}_{stat}(\mathcal{D})$  be the set of capacity vectors for which there exists a static fractional routing satisfying all demand vectors of $\mathcal{D}$.   $\mathcal{U}_{stat}(\mathcal{D})$ is obviously a polyhedral set.  
 The mathematical programs \eqref{congestionPb}, \eqref{dualPb} and \eqref{suite:cong} can be considered in the same way: we only have to replace $\mathcal{U}(\mathcal{D})$ by   $\mathcal{U}_{stat}(\mathcal{D})$. All results stated above about the equivalence of \eqref{congestionPb}, \eqref{dualPb} and \eqref{suite:cong} still hold in the static case.  Consider an instance from \cite{Maggs97,BARTAL199919}  for which $\frac{cong_{stat-frac}}{cong_{dyn-frac}}$ is $\Omega(\log n)$. $cong_{stat-frac}$ is computed from \eqref{suite:cong}. Then there is at least one vector $\lambda^{stat} \geq 0$ and one vector $u^{stat} \in \mathcal{U}_{stat}(D)$ such that  $cong_{stat-frac} = \sum_{e \in E(G)} \lambda^{stat}_e u^{stat}_e$ and $\sum_{e \in E(G)} \lambda^{stat}_e c_e = 1$. This implies that  $u^{stat}$ is an optimal solution of the linear problem where we minimize  $\sum_{e \in E(G)} \lambda^{stat}_e u_e$  under the condition $u \in \mathcal{\mathcal{U}}_{stat}(D)$. We consequently have $cong_{stat-frac} = Lin_{stat-frac}$ for the considered instance.\\
 Let $u' \in \mathcal{U}(\mathcal{D})$ be an optimal solution minimizing the linear cost $\sum_{e \in E(G)} \lambda^{stat}_e u_e$  under dynamic routing. In other words, $Lin_{dyn-frac} = \sum_{e \in E(G)} \lambda^{stat}_e u'_e$  when the coefficients of the objective function are $\lambda^{stat}$.
 Moreover, we know from \eqref{suite:cong}  that $cong_{dyn-frac}$ is obtained by maximizing through $\lambda$, implying that $cong_{dyn-frac} \geq \sum_{e \in E(G)} \lambda^{stat}_e u'_e = Lin_{dyn-frac}$.
 Using that $\frac{cong_{stat-frac}}{cong_{dyn-frac}}$ is $\Omega(\log n)$, we get that $\frac{{Lin}_{stat-frac}}{{Lin}_{dyn-frac}}$ is $\Omega(\log n)$ for the same instance where the linear objective function is defined through $\lambda^{stat}$.
 \end{proof}
 }
\section{Restriction to a constant number of given paths per commodity}

First, observe that in the proof of Lemma \ref{lemmaInapproximability}, the minimum congestion instances built there are such that  some commodities can be routed along many paths. For example, in graph $G_1$ (Figure \ref{fig:g1}), commodity $h_0$ (between $s$ and $t$) can use up to $m$ paths.  Second, consider an instance of the minimum congestion problem where only one path is given for each commodity. Then computing the minimum congestion is  easy  since we only have to compute $\max\limits_{d \in \mathcal{D}} \sum\limits_{h \in \mathcal{H}_e} d_h$ where $\mathcal{H}_e$ denotes the set of commodities routed through $e$. The congestion is just given by $\max\limits_{e \in E(G)} \frac{1}{c_e} \max\limits_{d \in \mathcal{D}} \sum\limits_{h \in \mathcal{H}_e} d_h$.
Combining these two observations, one can  wonder whether the difficulty of the congestion problem is simply due to the number of possible paths that can be used by each commodity. We will show that the problem is still difficult even if each commodity can be routed along at most two fixed given paths.  
\begin{figure}
\center
\includegraphics[width=13.5cm,height = 4cm]{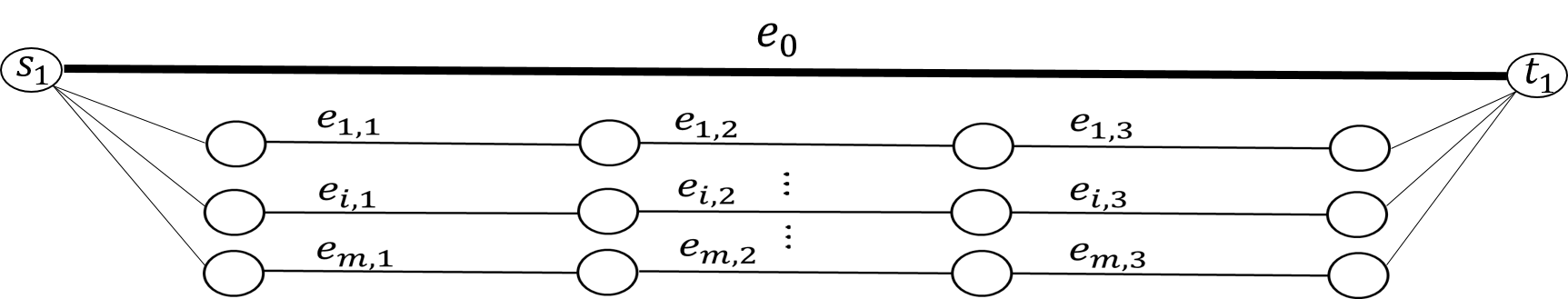}
\caption{$G'$}
\label{fig:gprime}
\end{figure}

\begin{theorem}
Unless $P=NP$, for some positive constant $k$, minimum congestion cannot be approximated within a
ratio $k$ even if each commodity can be routed along at most two given paths.  
\end{theorem}
\begin{proof}

The proof is a simple modification of the proof of Lemma \ref{lemmaInapproximability} (case $\gamma = 1$). We are  going to slightly modify graph $G_1$  in such a way that at most $2$ paths are allowed for  each commodity. 
Given a 3-SAT formula $\varphi$ with $m$ \textcolor{black}{clauses},  we construct $G',\mathcal{H}', \mathcal{D}'$ as follows.  We first create two nodes $s_1$ and $t_1$ and an edge $e_0$ between $s_1$ and $t_1$ of capacity $m \rho$ ($\rho$ is the constant in PCP theorem). 
Then for each clause index $i = 1,...,m$, as in Lemma \ref{lemmaInapproximability}, we create $3$ consecutive edges $e_{i,j}$ ($j = 1,2,3$) such that $t(e_{i,j}) = s(e_{i,j + 1})$ and a commodity $h_{i,j}$ between $s(e_{i,j}) $ and $t(e_{i,j})$ that is allowed to be routed only through  $e_{i,j}$. We also add one edge between $s(e_{i,1 })$ and $s_1$ and one edge connecting $t_1$ and $t(e_{i,3 })$ of infinite capacity and a commodity $h_{i,0}$ between $s(e_{i,1 })$ and $t(e_{i,3 })$ with a demand $d_{h_{i,0}} = 1 $. $h_{i,0}$ is allowed to be routed only through the path $P_i$ containing the edges $(e_{i,1} , e_{i,2} , e_{i,3} )$ and the path going through $s_1$, $e_0$ and $t_1$ (see Figure \ref{fig:gprime}).
We consider auxiliary variables  $\xi_{l }$ for  each literal $l$.  We add constraints $\xi_{l } + \xi_{ \lnot l } = 1 $  and   $d_{h_{i,j} } = \xi_{l_{i,j} }$. 

If $val(\varphi) < \rho$ there are at least $ m( 1 - \rho )$ commodities $h_{i,0}$ that can be routed on the paths $P_i$ and the  remaining $ m \rho $ can  be routed on the edge $e_0$. This implies that each extreme point of $\mathcal{D}'$ can be routed with congestion $\leq 1$. Notice that  the observation made in the proof of Lemma \ref{lemmaInapproximability} about extreme points is still valid here: extreme points corresponds to $0-1$ values of the variables $\xi_{l}$.

If $val(\varphi) = 1$, then there is a cut and a demand vector $d$ (corresponding to the truth assignment  satisfying $\varphi$) such that the capacity of the cut is $m \rho + m$ and the demand that needs to cross the cut is $ 2 m$. There is consequently at least one edge of congestion greater than or equal to $\frac{2 m}{  (1+\rho) m } = \frac{2}{1+\rho}$. By taking $k <\frac{2}{1+\rho}$ we get the wanted result.
\end{proof}
Finally, observe that the result above can also be stated for the linear cost case using again the Lagrange based reduction of the previous section. 

\begin{corollary}
\textcolor{black}{Unless P = NP,} for some positive constant $k$, robust network design with linear costs is difficult to approximate within a
ratio $k$ even if each commodity can be routed along at most two given paths.  
\end{corollary}

\textcolor{black}{
\section*{Acknowledgments}
The authors are grateful to anonymous referees for their helpful
suggestions that greatly improved the quality of the paper. We also point out that Theorem \ref{the:gap} was added following a suggestion of one of the reviewers who was wondering whether the reduction of Section \ref{sec:lagrange}
could be used to prove the theorem.} 

\newpage

\bibliographystyle{elsarticle-num}
\bibliography{bibliography} 

\end{document}